\documentclass[12pt, a4paper, onecolumn, accepted=2022-06-01]{quantumarticle}
\pdfoutput=1

\usepackage{amsfonts}
\usepackage{amsthm}
\usepackage{amsmath}
\usepackage{amscd}
\usepackage[latin2]{inputenc}
\usepackage{bbold}
\usepackage{t1enc}
\usepackage[mathscr]{eucal}
\usepackage{indentfirst}
\usepackage{graphicx}
\usepackage{graphics}
\usepackage{pict2e}

\numberwithin{equation}{section}

\theoremstyle{plain}
\newtheorem{Th}{Theorem}[section]
\newtheorem{Lemma}[Th]{Lemma}
\newtheorem{Cor}[Th]{Corollary}

\newtheorem{Pro}[Th]{Proposition}

 \theoremstyle{definition}

\newtheorem{?}[Th]{Problem}
\DeclareMathOperator{\infstor}{inf.\! stor}
\DeclareMathOperator{\asymm}{asymm}
\DeclareMathOperator{\lindim}{lin.\! dim}
\DeclareMathOperator{\affdim}{aff.\! dim}
\DeclareMathOperator{\signdim}{sign.\! dim}
\DeclareMathOperator{\ran}{ran}
\DeclareMathOperator{\Info}{Info}
\newcommand{\R}{\mathbb{R}}

\newcommand{\la}{\lambda}

\newcommand{\C}{\mathbb{C}}

\newcommand{\tr}{\operatorname{tr}}

\begin{document}

\title[Classical simulations of communication channels]
{Classical simulations of communication channels}

\author[P. E. Frenkel]{P\'eter E. Frenkel}
 
\address{E\"{o}tv\"{o}s Lor\'{a}nd University,
  P\'{a}zm\'{a}ny P\'{e}ter s\'{e}t\'{a}ny 1/C, Budapest, 1117 Hungary \\ and R\'enyi Institute,  Budapest, Re\'altanoda u.\ 13-15, 1053 Hungary}
\email{frenkelp265@gmail.com}

\thanks{Research  partially supported by MTA R\'enyi
``Lend\"ulet'' Groups and Graphs Research Group, 
 by  ERC Consolidator Grant 648017 and by NKFIH  grants K 124152 and KKP 139502.}



{\it In memoriam Katalin Marton}

\begin{abstract}\bf We investigate whether certain non-classical communication channels can be simulated by a classical channel with a  given number of states and a given `amount' of noise.  It is proved that any  noisy quantum channel can be simulated by a corresponding classical channel with `the same amount' of noise. Classical simulations of general probabilistic channels are also studied.\rm
\end{abstract}

\maketitle

\section*{Introduction}
A communication protocol with $l$ possible inputs and $k$ possible outputs  can be described by a \emph{transition matrix} $A=(a_{ij})\in [0,1]^{k\times l}$, where $a_{ij}$ is the conditional probability of output $i$  if the input is $j$. This is a \emph{stochastic} matrix
: for all $j$, we have $\sum_{i=1}^ka_{ij}=1$. A \emph{communication channel} can be described  by the set of transition matrices that it affords.  Channel Q \emph{can be simulated} by channel C  if all transition matrices afforded by Q are convex combinations of transition matrices afforded by C.  Such convex combinations occur naturally in information theory; they correspond to the sender and receiver having access to (unlimited)  shared randomness. The relation `can be simulated by' is obviously reflexive and transitive.  Two channels are \emph{equivalent} if each can be simulated by the other.

The \emph{classical channel with $n$ states} affords stochastic 0-1 matrices with at most $n$ nonzero rows.   The \emph{quantum channel of level $n$} affords channel matrices of the form $(\tr E_i\rho_j)$, where $\rho_1, \dots, \rho_l\in M_n(\C)$ are \emph{density matrices},  
and $E_1, \dots, E_k\in M_n(\C)$ is a \emph{positive operator valued measure (POVM)}. 
  It is easy to see that the classical channel with $n$ states can be simulated by the quantum channel of level $n$. By \cite[Theorem 3]{FW} of Weiner and the present author, the converse also holds.  The present paper is about variants of this theorem   for  general probabilistic channels (Section~\ref{GPT}) and for noisy quantum channels (Section~\ref{quantum}).  In Section~\ref{noiselesssimul}, we discuss noiseless classical simulations of noisy channels.  Section~\ref{future} contains an open problem tentatively linking classical simulations of quantum channels to the more traditional way of comparing  efficiency of classical and quantum communication, involving von Neumann entropy, mutual information and Holevo's inequality. The reader who is  interested in  quantum information theory but not in general probabilistic theory can safely skip Section~\ref{GPT}. 

\bigskip\noindent
\bf Notations and terminology.
\rm The set $\{1,\dots, k\}$ is denoted by  $[k]$.  For a real number $a$, we write  $a_+=\max (a, 0)$. The indicator of an event $A$ is written $\mathbb 1(A)$. A \emph{convex body} is a  convex compact set with nonempty interior.

 A matrix is \it stochastic
\rm if all entries are nonnegative reals and each column sums to 1.  The set of $n$-square matrices with complex entries is written $M_n(\C)$. The identity matrix is $\bf 1$.
A complex  matrix $A$  is \it psdh \rm if it is positive semi-definite Hermitian, written $A\ge 0$. A \emph{positive operator valued measure (POVM)} is a sequence $E_1$, \dots, $E_k$ of psdh matrices summing to $\bf 1$.
A \it density matrix \rm is a psdh matrix with trace 1.


For $0\le\delta\le 1$, the \emph{$\delta$-noisy classical channel with $n$ states} affords transition matrices of the form $EX\in [0,1]^{k\times l}$, where  $E\in \{0,1\}^{k\times n}$ is a stochastic 0-1 matrix and $X$ 
  is an $n\times l$ stochastic matrix with each column containing $n-1$ entries equal to $\delta/n$. Here  $E$ can be interpreted as  a classical decoding map $[n]\to[k]$, and the columns of $X$ can be interpreted as extremal $\delta$-noisy classical states. The presence of noise impedes the exact transmission of pure states, the pure state chosen by the sender is transmitted unchanged with probability $1-(n-1)\delta/n$ but turns into each one  of the other $n-1$ pure states with probability $\delta/n$.
Note that the 0-noisy classical channel with $n$ states is  the same as the classical channel with $n$ states defined previously. 

Following the  terminology of \cite{A, DC}, the \emph{signalling dimension} $\signdim \mathrm Q$ of a  channel $\mathrm Q$ is the smallest positive integer $n$ such that  $\mathrm Q$ can be simulated by the (noiseless) classical channel with $n$ states.

\section{General probabilistic theory}\label{GPT}
Let $S$ be  a convex body in a finite dimensional real affine space. Let $E$ be the cone of \emph{effects}, i.e., affine linear functions $e:S\to[0,\infty)$.  A \emph{partition of unity} is   a  sequence $e_1, \dots, e_k\in E$ of effects such that $e_1+\dots +e_k=1$  (the constant 1 function). The \emph{channel with state space $S$}  affords transition matrices of the form $(e_i(x_j))\in[0,1]^{k\times l}$, where $x_1, \dots, x_l\in S$, and $e_1$, \dots, $e_k$ is a partition of unity.

\subsection{Signalling dimension vs.\ information storability}
Following terminology introduced in \cite{A}, the \emph{signalling dimension} $\signdim S$ of $S$ is the signalling dimension of the channel with state space $S$, i.e., the smallest positive integer $n$ such that the channel with state space $S$ can be simulated by the classical channel with $n$ states.  By \cite[Theorem 3]{FW} mentioned in the Introduction, the signalling dimension of the set of $n$-square density matrices is $n$.

Calculating, or even efficiently estimating the signalling dimension of a given convex body seems to be a   difficult problem, and strong general theorems are yet to be searched for. In this section, we  start with weak  general results and work our way towards deeper results for   special cases.

The \emph{affine dimension} $\affdim S$ of $S$ is the minimal dimension of an affine space containing $S$.  Adding 1, we get the  \emph{linear dimension} $\lindim S$ of $S$, i.e., the dimension of the vector space of affine linear functions on $S$.  For example, the affine dimension of the set of $n$-square density matrices is $n^2-1$, while its linear dimension is $n^2$.

A partition of unity is \emph{extremal} if it cannot be written as a convex combination of two partitions of unity in a nontrivial way. The nonzero effects appearing in an extremal partition of unity need not lie on  extremal rays of the cone $E$ of effects. When they do, a characterization of extremal partitions of unity  is given in \cite[Theorem 2]{A}. We now give a necessary condition of extremality for a general partition of unity. Although  this
is implicitly contained  in the paper cited above (see the proof given there), we include  a proof.

\begin{Pro}\label{rang}  The nonzero effects in an extremal partition of unity are linearly independent. Thus, their number is $\le$ the linear dimension of $S$.
\end{Pro}

\begin{proof} Let $e_1$, \dots, $e_k$ be an extremal partition of unity. If $\la_1 e_1+\dots +\la_k e_k=0$ and $|\epsilon|\le 1/\max \{|\la_i|:  \la_i\ne 0\}$, then $(1\pm\epsilon\la_1)e_1$, \dots, $(1\pm\epsilon\la_k)e_k$ is also a partition of unity, which must coincide with $e_1$, \dots, $e_k$ because of extremality. Thus $\la_ie_i=0$ for all $i$.
\end{proof}

Consider the   transition matrix $A=(a_{ij})\in [0,1]^{k\times l}$ of some communication protocol, where $a_{ij}$ is the conditional probability of output $i\in[k]$ if the input was $j\in[l]$. Let us try to guess the input based on the output, using a  function $G:[k]\to [l]$.  If input $j$ occurs with  probability $q_j$, then the probability of success will be $$\sum_{j=1}^l q_j\sum _{i=1}^ka_{ij}\mathbb 1(G(i)=j)=\sum_{i=1}^kq_{G(i)}a_{i,G(i)}.$$  Choosing the best possible guessing function $G$, the probability of success is  $$\sum_{i=1}^k\max_j q_ja_{ij}.$$  Without any communication, the probability of successfully guessing the input, with the optimal  strategy, is $\max_j q_j$.  The ratio $\sum_{i=1}^k\max_j q_ja_{ij}/\max_j q_j$ is maximized when $q_j=1/l$ for all $j$, in which case it simplifies to  $\sum_{i=1}^k\max_j a_{ij}$.  Motivated by these considerations, and following \cite{MK} by Matsumoto and Kimura, the \emph{information storability} $\infstor S$ of $S$ is defined to be the maximum of $\sum_{i=1}^k\max_j a_{ij}$
over all transition matrices $(a_{ij})$  afforded by $S$, or, equivalently, the maximum of   $\sum_{i=1}^k\max_S e_{i}$  over all partitions of unity $e_1$, \dots, $e_k$.
When taking these maxima, it suffices to consider extremal partitions of unity. Then  Proposition~\ref{rang} and a  simple compactness argument shows  that these maxima are attained.

As a simple example, let $S=[0,1]$. Then we can choose the partition of unity $1=x+(1-x)$ to show that  $$\infstor S\ge \max_{0\le x\le 1}x+\max_{0\le x\le 1}(1-x)=1+1=2.$$ On the other hand, as any  affine  linear function on $S$ takes its maximum at 0 or 1, we have $$\sum_{i=1}^k\max_S e_i=\sum_{e_i(0)\ge e_i(1)} e_i(0)+\sum_{e_i(0)< e_i(1)} e_i(1)\le 1+1=2$$ for any partition of unity $e_1$, \dots, $e_n$ on $S$,  whence $\infstor S=2$.
This is the easiest special case of \cite[Theorems 1 and 4]{MK},  cited below in relation to  Theorem~\ref{signdim} and Proposition~\ref{MK}.

By \cite[Theorem 4]{MK},  $\infstor S\le \lindim S$. We  refine this inequality as follows.

\begin{Th}\label{signdim}
\begin{enumerate}
\item $\infstor S\le\signdim S \le \lindim S$.

\item
If  $\infstor S\le \affdim S$, then $\signdim S\le \affdim S$.

\end{enumerate}
\end{Th}

This theorem is closely related to \cite[Theorem 1(i)]{DC}.

\begin{proof}
(1)  Let $n=\signdim S$. Any transition matrix afforded by $S$ is a convex combination of transition matrices afforded by the classical channel with $n$ states. Such a matrix has $\le n$ nonzero rows and therefore sum of row-maxima $\le n$. This property is preserved  when taking convex combinations. This proves the first inequality.

Any transition matrix afforded by $S$ is a convex combination of transition matrices of the form  $(e_i(x_j))$, where $e_1$, \dots, $e_k$ is an \emph{extremal}  partition of unity, and $x_j\in S$. By Proposition~\ref{rang}, such a matrix has $\le\lindim S$ nonzero rows, and therefore is a convex combination  of matrices afforded by the classical channel with $\lindim S$ states. This proves the second inequality.

(2)
Let $\infstor S\le \affdim S=n$. Any transition matrix afforded by $S$ is a convex combination of matrices of the form $A=(a_{ij})\in [0,1]^{k\times l}$, where  $a_{ij}=e_i(x_j)$, $e_1$, \dots, $e_k$ is an \emph{extremal} partition of unity, and $x_j\in S$. We shall show that such an $A$ is always a convex combination of transition matrices afforded by the classical channel with $n$ states.  Using Proposition~\ref{rang}, we may assume that $k=n+1$.  Set $m_i=\max_S e_i\in [0,1]$ for each $i\in [k]$. 
Note that  
$\sum_{i=1}^k(1-m_i)\ge n+1-\infstor S\ge 1$.   Choose a probability distribution $p_1$, \dots, $p_k$ such that  $p_i\le 1-m_i$ for all $i$. Then $$p_i\le 1-a_{ij}=\sum_{i'\ne i} a_{i'j}$$  for all $i$ and $j$, and  $$\sum_{i\in T}p_i\le 1=\sum_{i=1}^k a_{ij}$$ for all $T\subseteq [k]$.

For any fixed $j$, put  supply $a_{ij}$ and demand $p_i$  at each node $i$  of the complete (but loopless)  graph on  $k$ nodes. Then, for the total supply  at the neighbors of any  subset $T\subseteq [k]$, we have $$\sum_{i\in N(T)}a_{ij}\ge \sum_{i\in T}p_i.$$  By the Supply--Demand Theorem~\cite[2.1.5.\ Corollary]{LoPlu}, the demands can be met: there exist stochastic column vectors $b_j(1)$, \dots, $b_j(k)$  such that  the $i$-th entry of $b_j(i)$ is zero for all $i$, and $\sum_{i=1}^kp_ib_j(i)$  is the $j$-th column of $A$. Now let $B(i)$ be the matrix with columns  $b_1(i)$, \dots, $b_l(i)$.  Then the $i$-th row of $B(i)$ is zero, so $B(i)$ has $\le k-1=n$ nonzero rows, so $B(i)$ is a convex  combination of transition matrices afforded by the classical channel with $n$ states.  Then so is $A$, since  $$A=\sum_{i=1}^k p_iB(i).$$
\end{proof}

For the remainder of this section, assume that $S$ is not just a point.  A \emph{chord} \rm  of $S$ is a segment $AB$ whose endpoints $A$ and $B$ belong to the boundary  of $S$.  We write $AOB$ for a chord $AB$ with a distinguished point $O$ on the chord. The convex body $S$ is \emph{centrally symmetric} if  there exists a point $O\in S$ such that for any chord $AOB$ of $S$, we have $|OA|=|OB|$ for the lengths of the segments $OA$ and $OB$.
The \emph{Minkowski measure of asymmetry}  $\asymm S$ of $S$ is  the smallest real number $m\ge 1$  such that  there exists a point $O\in S$ such that for any chord $AOB$ of $S$, we have $|OB|\le m|OA|$.

 By \cite[Theorem 1]{MK}  of Matsumoto and Kimura, the information storability  is related to the Minkowski measure of asymmetry  as follows.

\begin{Pro}\label{MK}  $\infstor S = \asymm S +1$
\end{Pro}
Although this is a known statement, we include the sketch of a geometric proof for the convenience of the reader.
\begin{proof}
$\le$:  
There exists a point $O\in S$ such that for any chord $AOB$ of $S$, we have $|OB|\le (\asymm S)|OA|$.
 Let $n=\asymm S +1$.  Then $e(x)\le ne(O)$ for all  $e\in E$ and $x\in S$, whence  $$\sum_{i=1}^k\max_S e_{i}\le n\sum_{i=1}^ke_i(O)=n$$   for all partitions of unity $e_1$, \dots, $e_k$.

$\ge$: Let $n=\infstor S$. Then   $\sum_{i=1}^k\max_S e_{i}\le n$   for all partitions of unity $e_1$, \dots, $e_k$.  When $k$ is the linear dimension of $S$, this tells us that for any simplex $\Delta$ containing $S$, there exists a point each of whose  barycentric coordinates with respect to $\Delta$  is at least $1/n$ times  the maximum value of that barycentric coordinate on $S$.
  Using Helly's theorem, we see that  there exists a point $O$ that divides  the distance between any two parallel supporting hyperplanes of $S$  in a ratio at least as equitable as $1:(n-1)$.  Then, for any chord $AOB$  of $S$ with $|AO|\le |OB|$, considering the supporting hyperplane  of $S$ at $A$  and the parallel supporting hyperplane, we get that $|OB|\le (n-1)|OA|$.  
\end{proof}

\begin{Cor} For the regular octahedron, we have $\asymm =1$, $\infstor =2$,  $\signdim=\affdim =3$, and $\lindim =4$.
\end{Cor}

\begin{proof}  The regular octahedron is centrally symmetric, which means that  $\asymm =1$.
By Proposition~\ref{MK}, we have $\infstor=\asymm +1 =2$. Obviously, $\affdim =3$ and $\lindim=\affdim +1=4$.

By Theorem~\ref{signdim}(2), we have $\signdim\le 3$.  To prove the converse inequality, let $$X=\begin{pmatrix} 1 & -1 &&&&\\&&1&-1&&\\&&&&1&-1\end{pmatrix}$$  be the matrix whose columns are the vertices of the octahedron (the entries not shown are zero).  Let $$V=\begin{pmatrix} 1&1&1\\1&-1&-1\\-1&1&-1\\-1&-1&1\end{pmatrix},$$ then $$VX=\begin{pmatrix} 1&-1&1&-1&1&-1\\1&-1&-1&1&-1&1\\-1&1&1&-1&-1&1\\-1&1&-1&1&1&-1\end{pmatrix}.$$ Adding 1 to each entry and dividing by 4, we get the stochastic matrix
$$A=\frac12\begin{pmatrix} 1&0&1&0&1&0\\1&0&0&1&0&1\\0&1&1&0&0&1\\0&1&0&1&1&0\end{pmatrix},$$
which is therefore a  transition matrix afforded by the octahedron.  Since any two rows of $A$ have an 1/2 at the same position, we have $$\sum_{1\le i<i'\le 4}\max_{1\le j\le 6}(a_{ij}+a_{i'j})= {4\choose 2} =6.$$ On the other hand, any $4\times 6$ transition matrix afforded by the classical channel with 2  states has at least $4-2=2$ zero rows, so the sum above would be 
$\le  {4\choose 2}-{{4-2}\choose 2}=5$  --- note that this is a special case of  \cite[inequality (3.6)]{FW}.   This inequality  is preserved under convex combinations.  Therefore, the octahedron cannot be simulated by the classical 2-state channel, hence its signalling dimension is (at least) 3.
\end{proof}

\subsection{Noisy balls}
If an  origin is chosen in $S$, and $0\le \delta \le 1$, then the \emph{$\delta$-noisy channel with state space $S$} affords the transition matrices $(e_i(x_j))$, where $e_1$, \dots, $e_k$  is a partition of unity and  $x_j\in (1-\delta)S$ for all $j$.  This is analogous to the partial depolarization channel in quantum information theory, cf.\ Subsection~\ref{depol}. Note that $e_i\ge 0$ is required on all of $S$.

It is easy to see that if $S'=f(S)$ is an affine image of $S$, then $S'$ can be simulated by $S$. If, in addition,    $O'=f(O)$, then $\delta$-noisy $S'$ can be simulated by $\delta$-noisy $S$. In particular, a  classical bit can be simulated by $S$ unless $S$ is just a point, and a $\delta$-noisy classical bit can be simulated by any $\delta$-noisy $S\ne\{O\}$  that   is symmetric with respect to $O$.

\begin{Th}\label{ball} Let $n$ be an even positive integer. Put $$S=\{x\in \R^d: \|x\|_{n/(n-1)}\le 1\},$$  the unit ball of the  $n/(n-1)$-norm. Let $0\le \delta \le 1$.
\begin{enumerate}
\item The $\delta$-noisy channel with state space $S$ can be simulated by the $\delta$-noisy classical channel with $n$ states.
\item The signalling dimension of $S$ is $\le n$.

\item For an ellipsoid of arbitrary affine dimension $\ge 1$, the signalling dimension  is $2$. A $\delta$-noisy ellipsoid can be simulated by a $\delta$-noisy classical bit.
\end{enumerate}
\end{Th}

The proof below is similar to that of \cite[Theorem 3]{FW}.  However, the mixed discriminant used there (and used in Section~\ref{quantum} of the present paper) must be replaced by a  different $n$-linear symmetric function $\{\cdot, \dots, \cdot\}$.

To introduce $\{\cdot, \dots, \cdot\}$,  we can think of an affine linear function $e:S\to \R$ as a formal sum of a number and a vector: $e=c+v\in\R\oplus \R^d=\R^{d+1}$, meaning that $e(x)=c+vx$ for $x\in S$, where $vx$ is the usual inner product. For an effect $e\in E$, the condition $e\ge 0$ translates to $\|v\|_n\le c$ because $$(n/(n-1))^{-1}+n^{-1}=1.$$  Given $e_1, \dots, e_n\in \R^{d+1}$, where $e_i=c_i+v_i$, we define $$\{e_1,\dots, e_n\}=c_1\cdots c_n-v_1\cdots v_n,$$ where $v_1\cdots v_n$ means that we take the coordinatewise product and then add up the coordinates (which is an $n$-linear generalization of the usual inner product). For $n=2$, $\{\cdot,\cdot\}$ is the Lorentzian indefinite symmetric bilinear product well known from the special theory of relativity.  For general $n$,   $\{\cdot,\dots,\cdot\}$ is symmetric,  multilinear and $\{1, \dots, 1\}=1$.  When $e_1, \dots, e_n\in E$, we have
$\{e_1,\dots, e_n\}\ge 0$ by repeated application of H\"older's inequality.  Further,  if $0\le e\le 1$ holds pointwise on $S$, then writing $e=c+v$ and $a=\|v\|_n$, we have $0\le a\le\min (c, 1-c)$ and therefore  \begin{align*}\{e,\dots, e\}=c^n-v^n\overset * =c^n-a^n=\\= (c-a)(c^{n-1}+c^{n-2}a+\dots +ca^{n-2}+a^{n-1})\le \\\le (c-a)(c+(1-c))^{n-1}=c-a=\min_{x\in S} e(x).\end{align*}  Note that the equality marked by a * holds because $n$ is even.

We are now ready to start the proof of Theorem~\ref{ball}.

\begin{proof}(1)
Let  $A\in[0,1]^{k\times l}$ be  a $\delta$-noisy transition matrix afforded by $S$, i.e., $$a_{ij}=e_i((1-\delta)x_j),$$ where $x_1, \dots, x_l\in S$, $e_i\in E$, and $e_1+\dots +e_k=1$. We shall prove that $A$ is a convex combination  of $\delta$-noisy $n$-state classical transition matrices.

If $e_i=c_i+v_i$ as before, then $c_1+\dots+c_k=1$, $v_1+\dots+v_k=0$, and  $$a_{ij}=c_i+(1-\delta)v_ix_j=\delta c_i+(1-\delta)e_i(x_j),$$ so $A=\delta C+(1-\delta) A'$, where $C$ is the matrix with entries $c_{ij}=c_i$  not depending  on $j$, and $A'$ is the matrix with entries $a'_{ij}=e_i(x_j)$.

For $I=(i_1,\dots, i_n)\in [k]^n$, put
\[p_I=\{e_{i_1},\dots, e_{i_n}\}.\] 
We have $p_I\ge 0$ for all $I$. Thus, we get a  measure $P$
on $[k]^n$ defined by $$P(T)=\sum_{I\in T}p_I.$$
Using the multilinearity of the bracket and the assumption that $e_1$, \dots, $e_k$ is a partition of unity, we see that
\[P([k]^n)=\{ 1,\dots,  1\}=\rm 1,\]
so $P$ is  a probability measure.

Let $D(I)$ be the matrix with entries $d(I)_{ij}=m(i,I)/n$  not depending on $j$, where $m(i,I)$ is the number of occurrences of $i$ in the sequence $I$. Then $\int D\mathrm d P=C$ because $$\int d_{ij}\mathrm d P=\sum_{I\in[k]^n}p_Im(i,I)/n=\{e_i, 1, \dots, 1\}=c_i=c_{ij}.$$

For any $R\subseteq [k]$, we may put $e_R=\sum_{i\in R} e_i$, and then  we have
\[P(R^n)=\{e_R, \dots, e_R\}\le\min_{x\in S} e_R(x)\le e_R(x_j)\] for all $j$ since $0\le e_R\le 1$.
The right hand side here is $A'_j(R)$, where $A'_j$ is the probability measure on $[k]$ given by the numbers $e_i(x_j)$. So we have
\[A'_j(R)\ge P(R^n)\qquad\textrm{ for all }R\subseteq [k].\]

Let us connect $I\in [k]^n$ to $i\in [k]$ by an edge if $i$ occurs in $I$. This gives us a bipartite graph. The neighborhood of any set $T\subseteq [k]^n$ is the set $R\subseteq [k]$ of indices occurring in some element of $T$. We always have $T\subseteq R^n$, whence
$$A'_j(R)\ge P(R^n)\ge P(T).$$ Thus, by the Supply--Demand Theorem~\cite[2.1.5.\ Corollary]{LoPlu}, and using the fact that both $A'_j$ and $P$ are probability measures, there exists a probability measure $P_j$ on $[k]^n\times [k]$ which is supported on the edges of the graph and has marginals $P$ and $A'_j$. Whenever $p_I\ne 0$, let $B'(I)$ be the $k\times l$ stochastic matrix whose $j$-th column is given by the conditional distribution  $P_j|I$ on $[k]$.  We have  $A'=\int B'\mathrm dP$.

 Now $B(I)=\delta D(I)+(1-\delta) B'(I)$  is a convex combination of $\delta$-noisy  $n$-state classical transition matrices,  and, in turn,  $A=\int B\mathrm d P$ is a convex combination of the $B(I)$, as desired.

\bigskip

(2)  Set $\delta=0$ in (1).

\bigskip

(3) The signalling dimension of an ellipsoid is the same as that of the Euclidean unit ball. This is $\le 2$ by (2), and is $\ge 2$ because the unit ball is not a point. The noisy claim follows from (1).
\end{proof}

\section{Noisy quantum channels}\label{quantum}
Let $$K\subseteq \Delta_n=\{(\xi_1,\dots, \xi_n): \xi_i\ge0 \; \textrm{ for all } \; i, \;  \xi_1+\dots +\xi_n=1\}$$
  be  a convex set of probability distributions that is invariant under all permutations of the $n$ coordinates.  The \emph{$K$-noisy classical channel}  affords  transition matrices of the form $EX\in [0,1]^{k\times l}$, where  $X\in K^l$  is an $n\times l$ matrix  with all columns in $K$, and $E$ is a  $k\times n$ stochastic 0-1 matrix. A density matrix is  \emph{$K$-noisy} if the sequence of its eigenvalues is in $K$. The \emph{$K$-noisy quantum channel} affords transition matrices of the form $(\tr E_i\rho_j)$, where $E_1$, \dots, $E_k$ is a POVM and $\rho_j$ is a $K$-noisy density matrix for $j=1, \dots, l$.

It is easy to see that  the $K$-noisy classical channel can be simulated by  the $K$-noisy quantum channel. Our goal is to prove the converse, which is a far-reaching generalization of \cite[Theorem 3]{FW}  mentioned in the Introduction.

In fact, we may generalize further. Let $K_j\subseteq \Delta_n$   ($j=1, \dots, l$)  be   convex sets, each of them  invariant under all permutations of the $n$ coordinates.  The \emph{$(K_1, \dots, K_l)$-noisy classical channel}  affords  transition matrices of the form $EX\in [0,1]^{k\times l}$, where  $X\in K_1\times\dots\times K_l$  is an $n\times l$ matrix  with $j$-th column in $K_j$, and $E$ is a  $k\times n$ stochastic 0-1 matrix.  The \emph{$(K_1, \dots, K_l)$-noisy quantum channel} affords transition matrices of the form $(\tr E_i\rho_j)$, where $E_1$, \dots, $E_k$ is a POVM and $\rho_j$ is a $K_j$-noisy density matrix for $j=1, \dots, l$.

It is easy to see that  the $(K_1, \dots, K_l)$-noisy classical channel can be simulated by  the $(K_1, \dots, K_l)$-noisy quantum channel. We shall prove the converse.

As in \cite{FW}, our  main tool is the \emph{mixed discriminant}, the unique symmetric  $n$-linear function $D$ on $M_n(\C)$ such that  $D(E, \dots, E)=\det E$ for
all  $E\in M_n(\C)$. Explicitly, if  $E_i=\left[e_i^1, \dots, e_i^n\right]$  are the columns, then \begin{equation}\label{*} D(E_1, \dots, E_n)=\frac1{n!}\sum_ {\pi\in \mathfrak S_n} \det\left[e_{\pi(1)}^1, \dots, e_{\pi(n)}^n\right].\end{equation}

We shall need the
following inequalities.

\begin{Lemma}\label{L1}  For $\la_1, \dots, \la_n\in [0,1]$ and $r=1,2, \dots, n$, we have
\begin{equation}\label{lambda}  \sum_{Q\subseteq[n]}(r-|Q|)_+\prod_{m\notin Q}\la_m\prod_{m\in Q}(1-\la_m)\le \la_1+\dots+\la_r,\end{equation} where $a_+=\max (a, 0)$.
\end{Lemma}

\begin{proof}We have $$(r-|Q|)_+\le \left|[r]\setminus Q\right|=\sum_{s=1}^r{\mathbb {1}} (s\notin Q)$$ for all $Q$. Thus, the left hand side of \eqref{lambda} is $$\le\sum_{s=1}^r\sum_{Q\subseteq [n]\setminus\{s\}}\prod_{m\notin Q}\la_m\prod_{m\in Q}(1-\la_m)=\sum_{s=1}^r\la_j\prod_{m\ne s}(\la_m+(1-\la_m))= \la_1+\dots+\la_r.$$
\end{proof}

\begin{Lemma}\label{L2}
For an $n$-square Hermitian matrix $0\le E\le\bf 1$ with eigenvalues $\la_1$,  \dots, $\la_n$, and $r=1,2, \dots, n$, we have  $$\sum_{q=0}^{r-1}  (r-q){n\choose q}  D(\underbrace{E, \dots, E}_{n-q}, \underbrace{{\bf 1}-E, \dots,{\bf 1}- E}_{q})\le\la_1+\dots+\la_r. $$
\end{Lemma}

\begin{proof}
  Since the spectrum and the mixed discriminant are both invariant under unitary conjugation, we may assume  that $E$ is a diagonal matrix. Then \eqref{*} reduces Lemma~\ref{L2} to Lemma~\ref{L1}.
\end{proof}

By Bapat's  \cite[Lemma 2(vi)]{B}, if  $E_1$, \dots, $E_n$ are all positive semidefinite Hermitian matrices, then  \begin{equation}\label{**}D(E_1, \dots, E_n)\ge 0.\end{equation}
Given a POVM $E_1, \dots, E_k\in M_n(\C)$, we define  \begin{equation}\label {p}  p_I=D(E_{i_1}, \dots, E_{i_n})\end{equation} for all $I=(i_1, \dots, i_n)\in [k]^n$. By multilinearity and \eqref{**}, this defines  a probability distribution on $[k]^n$.

\begin{Lemma}\label{L3}
If  $E_1, \dots, E_k\in M_n(\C)$ is  a POVM, $u_1$, \dots,  $u_k$  are real numbers, and $\la_1$, \dots, $\la_n$ are the eigenvalues of $E=\sum_{i=1}^k u_iE_i$,  then \begin{equation}\label{***}
\sum_{I\in[k]^n}  p_I\min\left\{\sum_{m\in S}u_{i_m}: S\subseteq [n], |S|=r\right\}\le \la_1+\dots+\la_r\end{equation} for all $r=1,2,\dots, n$.
\end{Lemma}

\begin{proof}
We may assume that all $u_i\ge 0$ because adding $u$ to all $u_i$ adds $ru$ to both sides of \eqref{***}.  We may assume $u_1\ge\dots \ge u_k$. Put $u_{k+1}=0$.  Write  $E=\sum_{i=1}^kv_iF_i$, where $v_i =u_i-u_{i+1}$ and $F_i=E_1+\dots+E_i$.

Let $\sigma_i$ be the sum of the $r$ smallest eigenvalues of $F_i$. Then \begin{equation}\label{1} \sum_{i=1}^kv_i\sigma_i\le\la_1+\dots+\la _r.\end{equation}
As $0\le F_i\le\bf 1$, we have \begin{equation}\label{2}  \sum_{q=0}^{r-1}   (r-q)\binom{n}{q} D(\underbrace{F_i, \dots, F_i}_{n-q}, \underbrace{{\bf 1}-F_i, \dots,{\bf 1}- F_i}_{q})\le  \sigma_i\end{equation}  for all $i$, by Lemma~\ref{L2}.

On the other hand, since  $u_i=v_i+\dots +v_k$, we have $$\min\left\{\sum_{m\in S}u_{i_m}: S\subseteq [n], |S|=r\right\}=\sum_{i=1}^kv_i\left(r-|\{m\in [n]: i_m>i\}|\right)_+.$$ It remains to check that  \begin{align*}\sum_{I\in[k]^n}  p_I\left(r-|\{m\in [n]: i_m>i\}|\right)_+=\\=\sum_{q=0}^{r-1}  (r-q){n\choose q} D(\underbrace{F_i, \dots, F_i}_{n-q}, \underbrace{{\bf 1}-F_i, \dots,{\bf 1}- F_i}_{q})\end{align*} for all $i\in [k]$. This follows from \begin{align*}\sum\left(p_I:I\in [k]^n, |\{m\in [n]:i_m>i\}|=q\right)=\\={n\choose q} D(\underbrace{F_i, \dots, F_i}_{n-q}, \underbrace{{\bf 1}-F_i,\dots,{\bf 1}- F_i}_{q}),\end{align*} which is clear from the definitions of $p_I$ and $F_i$, and from the symmetry and multilinearity of $D$.
\end{proof}

We are ready for the main result of this paper.

\begin{Th}\label{main} The $(K_1, \dots, K_l)$-noisy quantum channel can be simulated by  the $(K_1, \dots, K_l)$-noisy classical channel.
In particular, the $K$-noisy quantum channel can be simulated by the $K$-noisy classical channel.
\end{Th}

\begin{proof}
It suffices to prove that  for any  POVM $E_1$, \dots, $E_k$, and any $K$-noisy density matrix $\rho$, there exist points $x_I=(x_{I,1}, \dots, x_{I,n})\in K$  for each  $I=(i_1, \dots, i_n)\in [k]^n$  such that \begin{equation}\label{fo}
\tr E_i\rho=\sum_{I\in [k]^n}p_I\sum(x_{I,m}:m\in [n], i_m=i)
\end{equation} for each $i\in [k]$.  Here the $p_I$ are defined as in \eqref{p}.

Let the eigenvalues of $\rho$ be $0\le\mu_1\le\dots\le \mu_n$; we have $\mu_1+\dots+\mu_n=1$.   Since $\rho$ is $K$-noisy, we have $\mu=(\mu_1, \dots, \mu_n)\in K$.   Since $K$ is convex and invariant with respect to permutations, any  convex combination of permutations of $\mu$ is in $K$. Thus, if $x\in [0,1]^n$ is a stochastic  vector, and any $r$ distinct coordinates of $x$ sum to $\ge  \mu_1+\dots +\mu_r$ for each $r=1,2, \dots, n$, then $x\in K$. If we require
\begin{itemize}
\item these $2^n$ inequalities for each $x_I$, together with

\item $x_{I,m}\ge 0$  for all $I$ and $m$, and

\item \eqref{fo} for all $i$,

\end{itemize}
then
each $x_I$ will be  a stochastic vector  since
setting $r=n$  yields  $$x_{I,1}+\dots+x_{I,n}\ge \mu_1+\dots+\mu_n=1,$$ while summing \eqref{fo} for  $i=1,2,\dots, k$  yields $$1=\sum_{I\in[k]^n} p_I(x_{I,1}+\dots+x_{I,n}).$$
Therefore, it suffices to prove that the system of $(2^n+n)k^n$ inequalities and $k$ equations above has  a solution.
By the well-known Farkas Lemma, this is equivalent to saying that  a linear combination of the  inequalities and equations in the system cannot lead to the contradictory inequality $0\ge 1$. That is, it suffices to prove that if  nonnegative numbers $w_{I,H}$  $(I\in [k]^n, H\subseteq [n])$  and real numbers $u_1$, \dots, $u_k$ satisfy
\begin{equation}\label{start}\sum(w_{I,H}:H\subseteq [n], H\ni m)\le p_I u_{i_m}\end{equation} for all  $I\in [k]^n$ and all $m\in [n]$, then \begin{equation}\label{cel}\sum_{I\in[k]^n}\sum_{H\subseteq [n]}w_{I,H}(\mu_1+\dots +\mu_{|H|})\le\sum_{i=1}^ku_i\tr E_i\rho.\end{equation}

Let $\la_1\le\dots \le\la_n$  be the eigenvalues of $u_1E_1+\dots+u_kE_k$.   By von Neumann's inequality, the  right hand side of \eqref{cel} is  $$\ge \la_1\mu_n+\dots+\la_n\mu_1.$$

The coefficient of any $\mu_t$  on the left hand side  of \eqref{cel} is  $$\sum_{I\in[k]^n}\sum_{|H|\ge t}w_{I,H},$$ so it suffices to prove that  $$\sum_{t=n-r+1}^n\sum_{I\in[k]^n}\sum_{|H|\ge t}w_{I,H}\le\la_1+\dots+\la_r$$ for $r=1,\dots, n$. In view of Lemma~\ref{L3}, this follows if  $$\sum_{t=n-r+1}^n\sum_{|H|\ge t}w_{I,H}\le p_I\sum_{m\in S}u_{i_m}$$  for all $I\in [k]^n$ and all  $S\subseteq [n]$ with $|S|=r$. This follows from \eqref{start}  and the fact that  $$\sum_{n-r<t\le|H|}1=(|H|+r-n)_+\le |S\cap H|=\sum_{m\in S\cap H} 1$$ for all $H,S\subseteq [n]$ with $|S|=r$.
\end{proof}

\section{Simulation of a noisy channel by a noiseless one}\label{noiselesssimul}
Given the $K$-noisy channel, we might try to determine its signalling dimension, i.e., simulate it by a noiseless classical channel with as few states as possible.  In view of Theorem~\ref{main}, it makes no difference whether the given channel is classical or quantum.
\begin{Th}\label{noiseless} The $K$-noisy classical (or, equivalently, quantum) channel can be simulated by the noiseless d-state classical (or, equivalently, level $d$ quantum) channel  if and only if we have \begin{equation}\label{binom} \mu_1+\dots+\mu_r\ge\binom rd\bigg/\binom nd\end{equation}  for all $\mu=(\mu_1\le \dots\le\mu _n)\in K$ and all integers $d\le r\le n$.
\end{Th}

\begin{proof}  `Only if':  Let $\mu=(\mu_1\le \dots\le\mu _n)\in K$. Let $A=(a_{ij})$ be an $n\times n!$  stochastic matrix whose  columns are the $n!$ permutations of $\mu$.  Then $A$ is a transition matrix afforded by  the $K$-noisy channel, thus also by the noiseless $d$-state channel.  By \cite[Section 3]{FW}, we then have $$\binom nr  (\mu_1+\dots+\mu_r)=\sum_{|S|=r}\min_{j\in[l]}\sum_{i\in S}a_{ij}
\ge \binom{n-d}{n-r}$$ for all $d\le r\le n$, which is equivalent to \eqref{binom}.

`If': Let $S$ be a uniform random $d$-element subset of $[n]$.  It suffices to prove that, for any $\mu\in K$, there is a random element $m$ of $S$ whose distribution is given by $\mathbb P(m=r)=\mu_r$ for all $r=1, \dots, n$. We may assume $\mu_1\le \dots\le\mu _n$. Let $\nu_r=\mathbb P(\max S=r)$, then $$\nu_1+\dots+\nu_r=\mathbb P(\max S\le r)=\binom rd\bigg/\binom nd \le \mu_1+\dots+\mu_r,$$ so $\mu$ is a convex combination of the permutations of $\nu$. But $\nu$ is the distribution of the greatest element of $S$, so each permutation of $\nu$ is the  distribution of an element of $S$, thus $\mu$ is the distribution of a random element of $S$, as claimed.
\end{proof}

For $0\le\delta\le 1$,  the \emph{$\delta$-noisy  quantum channel of level $n$} affords transition matrices of the form $(\tr E_i\rho_j)$, where $E_1, \dots, E_k\in M_n(\C)$ is a POVM and $\rho_1, \dots, \rho_l\in M_n(\C)$ are density matrices with all eigenvalues $\ge\delta/n$. 
This channel is equivalent to the  $\delta$-noisy classical channel with $n$ states. This is a special case of Theorem~\ref{main}. Alternatively, it can be shown by  combining ideas  from the proofs of Theorem~\ref{ball}(1) and \cite[Theorem 3]{FW}.
\begin{Cor}\label{cor} Let $0\le\delta\le 1$. The signalling dimension of the $\delta$-noisy $n$-state classical (or, equivalently, $n$-level quantum) channel is $\lceil (1-\delta)n +\delta\rceil$.
\end{Cor}

\begin{proof}
The $\delta$-noisy $n$-state classical  channel can
be simulated by the noiseless $d$-state classical  channel  if and only if we have \begin{equation}\label{convex}r\delta/n\ge\binom rd\bigg/\binom nd\end{equation}  for all integers $d\le r\le n-1$  --- note that both sides of \eqref{binom} are 1 for $r=n$. In inequality~\eqref{convex}, the left hand side is linear in $r$, while the right hand side is convex for $r=0, 1, \dots$. Also, the inequality holds for $r=0,1,\dots, d-1$.  Therefore, it holds  for all integers $d\le r\le n-1$ if and only if it holds for $r=n-1$, i.e., $(n-1)\delta/n\ge (n-d)/n$, or, equivalently, $d\ge (1-\delta)n +\delta$.
\end{proof}

\subsection{Partial replacer quantum channels}\label{depol}
The usual mathematical model for a  noisy quantum channel is given in terms of a completely positive trace-preserving map $\mathcal N:M_m(\C)\to M_n(\C)$. Let $\ran \mathcal N$ stand for the set of density matrices $\mathcal N(\sigma)\in M_n(\C)$, where $\sigma\in M_m(\C)$ is a density matrix. The channel affords transition matrices of the form $(\tr E_i\rho_j)\in [0,1]^{k\times l}$, where $E_1$, \dots, $E_k$ is a POVM in $M_n(\C)$ and each $\rho_j$ is contained in  $\ran \mathcal N$. The \emph{signalling dimension} $\signdim\mathcal N$ of $\mathcal N$ is the signalling dimension of this channel, i.e., the smallest $d$ such that the channel can be simulated by the noiseless classical channel with $d$ states. If the spectrum of every $\rho\in\ran\mathcal N$ is contained in a  given permutation-invariant set $K\subseteq\Delta_n$, then  every transition matrix afforded by $\mathcal N$ is also afforded by the $K$-noisy quantum channel, so we can  can use Theorem~\ref{main}  to show that $\mathcal N$ can be simulated by the $K$-noisy classical channel. Then Theorem~\ref{noiseless} can be used to give an upper bound on the signalling dimension of $\mathcal N$.

An important special case is given by \emph{partial replacer channels}. Let $m\le n$. We embed $M_m(\C)$ into $M_n(\C)$  as the set of matrices that are zero outside of the upper  left $m$-square block.  We fix a density matrix $\rho\in M_n(\C)$. The \emph{replacer channel} $\mathcal N_\rho: M_m(\C)\to M_n(\C)$ is given by $\mathcal N_\rho(X)=(\tr X)\rho$. Given $0\le\delta\le 1$, the  \emph{partial replacer channel}  $\mathcal  N_\rho(\delta): M_m(\C)\to M_n(\C)$ is given by $\mathcal N_\rho(\delta)(X)=(1-\delta)X+\delta(\tr X)\rho$.  In \cite[Theorem 3]{DC} by Doolittle and Chitambar, it is shown that \begin{equation}\label{dcbounds}\lceil(1-\delta) m+\delta\rceil\le\signdim \mathcal N_\rho(\delta)\le \min\{m, \lceil(1-\delta) m+1\rceil\},\end{equation} and the upper bound is tight for the \emph{partial erasure channel}  given by the  \emph{erasure flag} $\rho$ which has entry 1 at position $(m+1, m+1)$  and zero elsewhere. Note that the difference between the upper and the lower bound in \eqref{dcbounds} is at most 1.

We shall now prove that the lower bound is tight if $m=n$ and $\rho$ is sufficiently mixed, in particular, if $\rho=\mathbf 1/n$ is the \emph{maximally mixed state}, yielding the \emph{partial depolarization channel}  $$\mathcal N(\delta)(X)=(1-\delta)X+(\delta/n)(\tr X)\mathbf 1.$$

From now on, we let $m=n$.
 Let $d=\lceil (1-\delta)n+\delta\rceil$ stand for the lower bound in  \eqref{dcbounds}.
Let $\mu_1\le\dots\le\mu_n$ stand for the eigenvalues of a  fixed density matrix $\rho$.

\begin{Pro} \begin{enumerate}
\item If $\delta(\mu_1+\dots+\mu_r)\ge\binom rd /\binom nd$ holds for $r=d, \dots, n-1$, then $\signdim \mathcal N_\rho(\delta)=d$.
\item 
The partial depolarization channel is equivalent to the $\delta$-noisy classical channel with $n$ states. 
\item The signalling dimension of the partial depolarization channel is $d$.
\end{enumerate}
\end{Pro}

\begin{proof} (1)
The eigenvalues $\mu_1'\le\dots\le\mu_n'$ of $\mathcal N_\rho(\delta) (\sigma)=(1-\delta)\sigma+\delta\rho\ge \delta\rho$ satisfy $\mu_1'+\dots+\mu_r'\ge\delta(\mu_1+\dots+\mu_r)$ for any density matrix $\sigma\in M_m(\C)$ and any $r=1, \dots, n$. Thus, $\mu_1'+\dots+\mu_r'\ge \binom rd /\binom nd$ for $r=1, \dots, n-1$, but also, trivially, for $r=n$. The claim  now follows from Theorem~\ref{noiseless} together with the first inequality in \eqref{dcbounds}.

(2)  The range $\ran \mathcal N(\delta) $  is the set of density matrices with all eigenvalues $\ge\delta/n$, so the  claim follows from Theorem~\ref{main}. 

(3)  follows from (2) together with Corollary~\ref{cor}.
\end{proof}

\section{Future research}\label{future} It is well known that quantum communication can outperform classical communication if entanglement is used cleverly. On the other hand, in certain scenarios not involving entanglement, it can be proved that passing from classical to  quantum cannot increase efficiency.

 A fundamental result in this direction is the Holevo bound \cite{H}
   which we now recall.  For any stochastic matrix $A=(a_{ij})\in[0,1]^{k\times l}$  and input probabilities $q_j\ge 0$ $(j=1, \dots, l)$ summing to 1,  we define the \emph{mutual information} $$\Info(A, q)=H(j)+H(i)-H(i,j).$$ Here $H$ stands for the Shannon entropy of a random variable, and 
 the joint distribution of the random pair $(i,j)$ is given by the probabilities $q_ja_{ij}$. Now,   for any density matrices $\rho_j\in M_n(\C)$  and any POVM $E_1, \dots, E_k\in M_n(\C)$, the 
 Holevo inequality  reads \begin{equation}\label{holevo}\Info(A,q)\le\chi,\end{equation}  where $a_{ij}=\tr E_i\rho_j$ and the  Holevo quantity $\chi$ is defined by $$\chi=S\left(\sum_{j=1}^l q_j\rho_j\right)-\sum_{j=1}^l q_jS(\rho_j),$$ where $S$ is von Neumann entropy, i.e., the Shannon entropy of the spectrum.  If all $\rho_j$ with $q_j>0$ commute, then a  POVM $E_1$, \dots $E_k$ can be found so that equality holds in \eqref{holevo}. Otherwise, the inequality is strict for any POVM.

 Another result in the above mentioned direction is \cite[Theorem 3]{FW}: the $n$-level quantum channel can be simulated by the $n$-state classical channel.

  It would be nice to unify these two results. Let  a probability distribution $q_1$, \dots, $q_l$ 
   be given.
    Can every quantum transition matrix $A=(a_{ij})=(\tr E_i\rho_j)\in[0,1]^{k\times l}$, where $E_1, \dots, E_k\in M_n(\C)$ is a POVM, and $\rho_1, \dots, \rho_l\in M_n(\C)$ are density matrices, be written as a convex combination $A=\sum
    p_IA_I$ of stochastic matrices $A_I$, each  with $\le n$ nonzero rows, and each satisfying $\Info(A_I, q)\le\chi $ ?
      Can the proof of Theorem~\ref{main} be modified to yield this result and thus, maybe,  a new proof  of Holevo's inequality? 

\bigskip

\noindent
{\bf Acknowledgement.}
I am grateful to Mih\'aly Weiner for useful conversations.


\begin{thebibliography}{99}


\bibitem{B}  R.\! B.\!  Bapat:  Mixed discriminants of positive semidefinite matrices. \it Linear Algebra Appl.\ \bf 126 \rm (1989), 107--124. \href{https://doi.org/10.1016/0024-3795(89)90009-8}{https:/\!/doi.org/10.1016/0024-3795(89)90009-8}
\bibitem{A} Michele Dall'Arno, Sarah Brandsen, Alessandro Tosini, Francesco Buscemi, and Vlatko Vedral: No-Hypersignaling Principle, Phys.\ Rev.\ Lett.\ 119 (2017), 020401.  \href{https://doi.org/10.1103/PhysRevLett.119.020401}{https:/\!/doi.org/10.1103/PhysRevLett.119.020401}
\bibitem{DC} Brian Doolittle, Eric Chitambar: Certifying the Classical Simulation Cost of a Quantum Channel, 	Phys.\ Rev.\ Research 3, 043073. \href{https://doi.org/10.1103/PhysRevResearch.3.043073}{https:/\!/doi.org/10.1103/PhysRevResearch.3.043073}
\bibitem{FW}P.\! E.\! Frenkel and M.\! Weiner: Classical information storage in an $n$-level quantum system, Communications in Mathematical Physics  340 (2015), 563--574. \href{https://doi.org/10.1007/s00220-015-2463-0}{https:/\!/doi.org/10.1007/s00220-015-2463-0}

\bibitem{H} A.\! S.\! Holevo: Bounds for the Quantity of Information Transmitted by a Quantum Communication Channel,  Probl.\ Peredachi Inf., 9:3 (1973), 3--11; Problems Inform.\ Transmission,  9:3 (1973), 177--183.




\bibitem{LoPlu}
L.\ Lov\'asz  and M.\ D.\ Plummer: Matching Theory. North-Holland, 1986.

\bibitem{MK} Keiji Matsumoto, Gen Kimura: Information-induced asymmetry of state space in view of general probabilistic theories, 
\href{https://doi.org/10.48550/arXiv.1802.01162}{
https:/\!/doi.org/10.48550/arXiv.1802.01162}
 \end{thebibliography}
\end{document}